\documentclass[a4paper,envcountsame,12pt]{llncs}
\usepackage[utf8]{inputenc}
\usepackage[english]{babel}
\usepackage[T1]{fontenc}
\usepackage[colorlinks=true,citecolor=black,linkcolor=black,urlcolor=black]{hyperref}	
\usepackage{amssymb,amsmath}
\usepackage{tikz}
\usetikzlibrary{snakes}
\usetikzlibrary{fit}
\usepackage{a4wide}
\usepackage[color=green!30]{todonotes}
\usepackage{xspace}
\usepackage{mathrsfs}

\usepackage{longtable}
\usepackage{booktabs}

\usepackage{aliascnt}	

\renewcommand\leq\leqslant
\renewcommand\geq\geqslant

\newtheorem{observation}[theorem]{Observation}
\newcommand{\NEQCLkl}{\textsc{Partial VC Dimension}\xspace}
\newcommand{\MaxNEQCL}{\textsc{Max Partial VC Dimension}\xspace}

\newcommand{\PBTC}{\textsc{Test Cover}\xspace}
\newcommand{\PBDS}{\textsc{Distinguishing Transversal}\xspace}
\newcommand{\MinPBDS}{\textsc{Min Distinguishing Transversal}\xspace}
\newcommand{\VCDIM}{\textsc{VC Dimension}\xspace}
\newcommand{\MaxVCDIM}{\textsc{Max VC Dimension}\xspace}

\newcommand{\MaxVC}{\textsc{Max Partial Vertex Cover}\xspace}
\newcommand{\MinVC}{\textsc{Min Vertex Cover}\xspace}

\newcommand{\apx}{\mathsf{APX}}

\newlength{\atextwidth}
\newcommand{\fullversion}[1]{#1}
\newcommand{\confversion}[1]{}

\usepackage{etoolbox}

\newcounter{Bew1}
\newcounter{Bew2}
\newcounter{Def1}

\newcommand{\problemdec}[3]{
  \vspace{2mm}
\noindent\fbox{
  \begin{minipage}{\atextwidth}
  \begin{tabular*}{\textwidth}{@{\extracolsep{\fill}}lr} #1 \\ \end{tabular*}
  {\bf{Input:}} #2  \\
  {\bf{Question:}} #3
  \end{minipage}
  }
  \vspace{2mm}
}

\newcommand{\problemopt}[3]{
  \vspace{2mm}
\noindent\fbox{
  \begin{minipage}{\atextwidth}
  #1 \\
  {\bf{Input:}} #2  \\
  {\bf{Output:}} #3
  \end{minipage}
  }
  \vspace{2mm}
}

\let\doendproof\endproof
\renewcommand\endproof{~\hfill\qed\doendproof}
\begin{document}

\fullversion{\title{Parameterized and Approximation Complexity of \NEQCLkl\thanks{A preliminary version of this work containing the approximation complexity results appeared in~\cite{confpaper}.}}}
\confversion{\title{On the Approximability of \NEQCLkl}}
\author{Cristina Bazgan \inst{1,\thanks{Institut Universitaire de France}} \and Florent Foucaud\inst{2} \and Florian Sikora \inst{1}}

\institute{Universit\'{e} Paris-Dauphine, PSL Research University, CNRS, LAMSADE, PARIS, FRANCE \\ \email{\{cristina.bazgan,florian.sikora\}@dauphine.fr} \and  Universit\'e Clermont Auvergne, LIMOS, CNRS, AUBIERE, FRANCE\\ \email{florent.foucaud@gmail.com}}

\maketitle

\begin{abstract}
We introduce the problem \NEQCLkl that asks, given a hypergraph $H=(X,E)$ and integers $k$ and $\ell$, whether one can select a set $C\subseteq X$ of $k$ vertices of $H$ such that the set $\{e\cap C, e\in E\}$ of distinct hyperedge-intersections with $C$ has size at least $\ell$. 
The sets $e\cap C$ define equivalence classes over $E$. \NEQCLkl is a generalization of \VCDIM, which corresponds to the case $\ell=2^k$, and of \PBDS, which corresponds to the case $\ell=|E|$ (the latter is also known as \PBTC in the dual hypergraph). 
We also introduce the associated fixed-cardinality maximization problem \MaxNEQCL that aims at maximizing the number of equivalence classes induced by a solution set of $k$ vertices. 
\fullversion{We study the algorithmic complexity of \NEQCLkl and \MaxNEQCL both on general hypergraphs and on more restricted instances, in particular, neighborhood hypergraphs of graphs.}
\confversion{We study the approximation complexity of \MaxNEQCL on general hypergraphs and on more restricted instances, in particular, neighborhood hypergraphs of graphs.}
\end{abstract}


\section{Introduction}

We study identification problems in discrete structures. 
Consider a hypergraph (or set system) $H=(X,E)$, where $X$ is the vertex set and $E$ is a collection of hyperedges, that is, subsets of $X$. 
Given a subset $C\subseteq X$ of vertices, we say that two hyperedges of $E$ are \emph{distinguished} (or \emph{separated}) by $C$ if some element in $C$ belongs to exactly one of the two hyperedges. 
In this setting, one can tell apart the two distinguished hyperedges simply by comparing their intersections with $C$. 
Following this viewpoint, one may say that two hyperedges are related if they have the same intersection with $C$. This is clearly an equivalence relation, and one may determine the collection of equivalence classes induced by $C$: each such class corresponds to its own subset of $C$. 
Any two hyperedges belonging to distinct equivalence classes are then distinguished by $C$. We call these classes \emph{neighborhood equivalence classes}.
In general, one naturally seeks to distinguish as many pairs of hyperedges as possible, using a small set $C$.

It is a well-studied setting to ask for a maximum-size set $C$ such that $C$ induces all possible $2^{|C|}$ equivalence classes. 
In this case, $C$ is said to be \emph{shattered}. 
The maximum size of a shattered set in a hypergraph $H$ is called its \emph{Vapnis-\v{C}ervonenkis dimension} (VC dimension for short). This notion, introduced by Vapnis and \v{C}ervonenkis~\cite{VC71} arose in the context of statistical learning theory as a measure of the structural complexity of the data. Hypergraphs of fixed VC dimension can be seen as special restrictions of \emph{covering arrays}, a type of combinatorial structure (see~\cite{covarr} for a survey). VC dimension has been widely used in discrete mathematics: see the references in the thesis~\cite{bousquet}. We have the following associated decision problem.

\smallskip
\problemdec{\VCDIM}{A hypergraph $H=(X,E)$, and an integer $k$.}{Is there a shattered set $C\subseteq X$ of size at least $k$ in $H$?}

The complexity of \VCDIM was studied in e.g.~\cite{CHKX06,DEF93,PY96}; it is a complete problem for the complexity class LOGNP defined in~\cite{PY96} (it is therefore a good candidate for an NP-intermediate problem). \VCDIM remains LOGNP-complete for \emph{neighborhood hypergraphs of graphs}~\cite{KKRUW97}: the \emph{neighborhood hypergraph} of $G$ has $V(G)$ as its vertex set, and the set of closed neighborhoods of vertices of $G$ as its hyperedge set.

In another setting, one wishes to distinguish \emph{all} pairs of hyperedges (in other words, each equivalence class must have size $1$) while minimizing the size of the solution set $C$. 
Following~\cite{HY14}, we call the associated decision problem, \PBDS. 

\smallskip
\problemdec{\PBDS}{A hypergraph $H=(X,E)$, and an integer $k$.}{Is there a set $C\subseteq X$ of size at most $k$ that induces $|E|$ distinct equivalence classes?}

There exists a rich literature about \PBDS. It was studied under different names, such as \textsc{Test Set} in Garey and Johnson's book~\cite[SP6]{GJ79}; other names include \PBTC~\cite{CGJMY16,CGJSY12,DHHHLRS03}, \textsc{Discriminating Code}~\cite{CCCHL08} or \textsc{Separating System}~\cite{BS07,R61}.\footnote{Technically speaking, in \textsc{Test Set}, \PBTC and \textsc{Separating System}, the goal is to distinguish the \emph{vertices} of a hypergraph using a set $C$ of \emph{hyperedges}, and in \textsc{Discriminating Code} the input is presented as a bipartite graph. 
Nevertheless, these formulations are equivalent to \PBDS by considering either the dual hypergraph of the input hypergraph $H=(X,E)$ (with vertex set $E$ and hyperedge set $X$, and hyperedge $x$ contains vertex $e$ in the dual if hyperedge $e$ contains vertex $x$ in $H$), or the bipartite incidence graph (defined over vertex set $X\cup E$, and where $x$ and $e$ are adjacent if they were incident in $H$).} A famous theorem of Bondy~\cite{B72} also implicitly studies this notion. \PBDS restricted to neighborhood hypergraphs of graphs is equivalent to the graph problem \textsc{Identifying Code} studied for example in~\cite{F13,KCL98}.

The goal of this paper is to introduce and study the problem \NEQCLkl, that generalizes both \PBDS and \VCDIM, and defined as follows.

\problemdec{\NEQCLkl}{A hypergraph $H=(X,E)$, and two integers $k$ and $\ell$.}{Is there a set $C\subseteq X$ of size $k$ that induces at least $\ell$ distinct equivalence classes?}




\NEQCLkl belongs to the category of \emph{partial} versions of common decision problems, in which, instead of satisfying the problem's constraint task for all elements (here, all $2^{k}$ equivalence classes), we ask whether we can satisfy a certain number, $\ell$, of these  constraints. 
See for example the papers~\fullversion{\cite{CMPS14,FLR11,KMR07}}\confversion{\cite{FLR11,KMR07}} that study some partial versions of standard decision problems, such as \textsc{Set Cover}\fullversion{, \textsc{Vertex Cover}} or \textsc{Dominating Set}.

When $\ell=|E|$, \NEQCLkl is precisely the problem \PBDS. When $\ell=2^k$, we have the problem \VCDIM. Hence, \NEQCLkl is NP-hard, even on many restricted classes. Indeed, \PBDS is NP-hard~\cite{GJ79}, even on hypergraphs where each vertex belongs to at most two hyperedges~\cite{DHHHLRS03}, or on neighborhood hypergraphs of graphs that are either: unit disk~\cite{MS09}, planar bipartite subcubic~\cite{F13}, interval~\cite{FMNPV15}, permutation~\cite{FMNPV15}, or split~\cite{F13}.
\MinPBDS cannot be approximated within a factor of $o(\log n)$ on hypergraphs of order $n$~\cite{DHHHLRS03}, even on hypergraphs without $4$-cycles in their bipartite incident graph~\cite{BLLPT15}, 
and on neighborhood hypergraphs of bipartite, co-bipartite or split graphs~\cite{F13}.

When $\ell=2^k$, \NEQCLkl is equivalent to \VCDIM and unlikely to be NP-hard (unless all problems in NP can be solved in quasi-polynomial time), since $|X|\leq 2^k$ and a simple brute-force algorithm has quasi-polynomial running time. 
Moreover, \VCDIM (and hence \NEQCLkl) is W[1]-complete when parameterized by $k$~\cite{DEF93}.

Recently, the authors in~\cite{BKMN15} introduced the notion of \emph{$(\alpha,\beta)$-set systems}, that is, hypergraphs where, for any set $S$ of vertices with $|S|\leq \alpha$, $S$ induces at most $\beta$ equivalence classes. Using this terminology, if a given hypergraph $H$ is an $(\alpha,\beta)$-set system, $(H,k,\ell)$ with $k=\alpha$ is a YES-instance of \NEQCLkl if and only if $\ell\leq\beta$.

We will also study the approximation complexity of the following fixed-cardinality maximization problem associated to \NEQCLkl.

\problemopt{\MaxNEQCL}{A hypergraph $H=(X,E)$, and an integer $k$.}{A set $C\subseteq X$ of size $k$ that maximizes the number of equivalence classes induced by $C$.}

Similar \emph{fixed-cardinality} versions of classic optimization problems such as \textsc{Set Cover}, \textsc{Dominating Set} or \textsc{Vertex Cover}, derived from the "partial" counterparts of the corresponding decision problems, have gained some attention in the recent years, \fullversion{see for example the papers~\cite{BKL12,C07,randomsep,CMPS14,feige,KMR07} and the survey~\cite{BEHM06}.}
\confversion{see for example~\cite{C07,KMR07,BEHM06}.}

\MaxNEQCL is clearly NP-hard since \NEQCLkl is NP-complete; other than that, its approximation complexity is completely unknown since it cannot be directly related to the one of approximating \MinPBDS or \MaxVCDIM (the minimization and maximization versions of \PBDS and \VCDIM, respectively).

\paragraph{Our results.} 
\fullversion{We first study the decision problem \NEQCLkl. In Section~\ref{sec:nat-param}, we show that \NEQCLkl parameterized by $k$ or $\ell$ is W[1]-complete. The W[1]-hardness holds even for neighborhood hypergraphs of bipartite, split or co-bipartite graphs. On the other hand, the problem is in FPT for neighborhood hypergraphs of graphs of bounded maximum degree.
In Section~\ref{sec:dual-param}, we study the dual parameterizations of \NEQCLkl. In particular, the problem is  W[1]-hard parameterized by   $|X|-k$.}
\fullversion{
We next turn our attention to the approximation complexity of \MaxNEQCL. 
}
\confversion{Our focus is on the approximation complexity of \MaxNEQCL.}
We give positive results in Section~\ref{sec:pos-approx}. We first provide polynomial-time approximation algorithms using the VC dimension, the maximum degree and the maximum edge-size of the input hypergraph. 
We apply these to obtain approximation ratios of the form $n^{\delta}$ (for $\delta<1$ a constant) in certain special cases, as well as a better approximation ratio but with exponential running time. 
For neighborhood hypergraphs of planar graphs, \MaxNEQCL admits an EPTAS (this is also shown for \MinPBDS). 
In Section~\ref{sec:neg-approx}, we give hardness results. 
We show that any $2$-approximation algorithm for \MaxNEQCL implies a $2$-approximation algorithm for \MaxVCDIM. 
Finally, we show that \MaxNEQCL admits no PTAS (unless P$=$NP), even for graphs of maximum degree at most~$7$.
\fullversion{We start the paper with some preliminaries in Section~\ref{sec:prelim} and conclude in Section~\ref{sec:conclu}.}

%
%
%
%
%
%

\section{Preliminaries}\label{sec:prelim}

\paragraph{Generalities.}

The following easy but important observation is used implicitly when studying \NEQCLkl.

\begin{observation}\label{obs:add-remove}
In a hypergraph $H=(X,E)$, given two sets $S$ and $S'$ with $S\subset S'\subseteq X$, $S'$ induces at least as many equivalence classes as $S$ in $H$.
\end{observation}

The following is an easy reformulation of our main problem.

\begin{observation}\label{obs:reform}
A set $S$ of vertices of a hypergraph $(X,E)$ induces $\ell$ equivalence classes if and only if there is a set $L\subseteq E$ of $\ell$ distinct hyperedges such that for any two distinct hyperedges of $L$, there is a vertex of $S$ that belongs to exactly one of them.
\end{observation}

\paragraph{Graphs versus hypergraphs.}

In many places, we study \NEQCLkl on neighborhood hypergraphs of graphs. In that case, one may naturally view \NEQCLkl as a problem on graphs (this is also usually done for the special cases \VCDIM, as in~\cite{bousquet,KKRUW97}, and \PBDS, as in~\cite{F13}). Indeed, it becomes equivalent to select a set $S$ of $k$ vertices of the graph. Each equivalence class induced by $S$ corresponds to a set $N[v]\cap S$ for some vertex $v$ (where $N[v]$ denotes the closed neighborhood of vertex $v$). Therefore, in some of our proofs dealing with neighborhood hypergraphs of graphs, we may (implicitly) treat \NEQCLkl as a graph problem for the simplicity of our exposition.

\paragraph{Twin-free hypergraphs.}

In a hypergraph $H$, we call two equal hyperedges \emph{twin hyperedges}. Similarly, two vertices belonging to the same set of hyperedges are \emph{twin vertices}.

Clearly, two twin hyperedges will always belong to the same neighborhood equivalence classes. Similarly, for any set $T$ of mutually twin vertices, there is no advantage in selecting more than one of the vertices in $T$ when building a solution set $C$.

\begin{observation}\label{obs}
Let $H=(X,E)$ be a hypergraph and let $H'=(X',E')$ be the hypergraph obtained from $H$ by deleting all but one of the hyperedges or vertices from each set of mutual twins. Then, for any set $C\subseteq X$, the equivalence classes induced by $C$ in $H$ are the same as those induced by $C\cap X'$ in $H'$.
\end{observation}

Therefore, since it is easy to detect twin hyperedges and vertices in an input hypergraph, in what follows, we will always restrict ourselves to hypergraphs without twins. We call such hypergraphs \emph{twin-free}.



\paragraph{Degree conditions.}

In a hypergraph $H$, the \emph{degree} of a vertex $x$ is the number of hyperedges it belongs to. The \emph{maximum degree} of $H$ is the maximum value of the degree of a vertex of $H$; we denote it by $\Delta(H)$.

The next theorem gives an upper bound on the number of neighborhood equivalence classes that can be induced when the degrees are bounded. 
\fullversion{We reproduce a short proof from the literature for completeness.}

\begin{theorem}[\cite{CGJMY16,DHHHLRS03,KCL98}]\label{thm:delta-LB}
Let $H=(X,E)$ be a hypergraph with maximum degree~$\Delta$ and let $C$ be a subset of $X$ of size~$k$. Then, $C$ cannot induce more than $\frac{k(\Delta+1)}{2}+1$ neighborhood equivalence classes.
\end{theorem}
\fullversion
{
\begin{proof}
Denote by $i_1$ the number of equivalence classes
whose members contain a unique element of $C$, and $i_2$, the number of equivalence classes
whose members contain at least two elements of $C$. We have $i_1\leq k$ and $i_2\leq
\tfrac{k\Delta-i_1}{2}$ since each element of $C$ is contained in at most $\Delta$ edges. Hence, the number of classes is at most $i_1+i_2+1\leq
k+\tfrac{k\Delta}{2}-\tfrac{k}{2}=\tfrac{k(\Delta+1)}{2}+1$.
\end{proof}
}



\paragraph{The Sauer-Shelah lemma.}

The following theorem is known as the Sauer-Shelah Lemma~\cite{S72,Sh72} (it is also
credited to Perles in~\cite{Sh72} and a weaker form was stated by Vapnik and \v{C}ervonenkis~\cite{VC71}). It is a fundamental tool in the study of the VC dimension.

\begin{theorem}[Sauer-Shelah Lemma~\cite{S72,Sh72}]\label{thm:sauer}
Let $H=(X,E)$ be a hypergraph with more than $\sum_{i=0}^{d-1}\binom{|X|}{i}$ distinct hyperedges. Then, $H$ has VC dimension at least $d$.
\end{theorem}

Theorem~\ref{thm:sauer} is known to be tight. Indeed, the system that consists of considering all subsets of $\{1,\ldots,n\}$ of cardinality at most $d-1$ has VC dimension equal to $d-1$. 
Though the original proofs of Theorem~\ref{thm:sauer} were non-constructive, Ajtai~\cite{A98} gave a constructive proof that yields a (randomized) polynomial-time algorithm, and an easier proof of this type can be found in Miccianio~\cite{M01}.

The following direct corollary of Theorem~\ref{thm:sauer} is observed for example in~\cite{BLLPT15}.

\begin{corollary}\label{cor:VCdim}
Let $H=(X,E)$ be a hypergraph with VC dimension at most $d$. Then, for any subset $X'\subseteq X$, there are at most $\sum_{i=0}^{d}\binom{|X'|}{i}\leq |X'|^d+1$ equivalence classes induced by $X'$.
\end{corollary}

\confversion{\paragraph{Approximation.}}
\fullversion{\paragraph{Definitions and basic results about computational complexity.}}

{
We now recall some standard definitions in parameterized complexity, but we refer to the book~\cite{DF13} for details.

A \emph{parameterized problem} is a decision problem together with a \emph{parameter}, that is, an integer depending on the instance. 

The class XP denotes the class of parameterized problems that can be solved in time $|I|^{f(k)}$, where $I$ is the instance, $k$ is the parameter, and $f$ is computable. The class FPT (for \emph{fixed-parameter tractable}) denotes the class of parameterized problems that can be solved in time $f(k)\cdot|I|^c$ for an instance $I$ of size $|I|$ with parameter $k$, where $f$ is a computable function and $c$ is a constant. 

An \emph{FPT reduction} between two parameterized problems $A$ and $B$ is a function mapping an instance $I$ of $A$ with parameter $k$ to an instance $f(I)$ of $B$ with parameter $g(k)$, where $f(I,k)$ and $g(k)$ are computable in FPT time with respect to parameter $k$, and where $I$ is a YES-instance of $A$ if and only if $f(I)$ is a YES-instance of $B$.
The \textsc{Weighted Circuit Satisfiability} (WCS) problem parameterized by $k$ takes as input a boolean circuit and an integer $k$, and decides if there is a satisfying assignment for this circuit with exactly $k$ input vertices of the circuit set to true. 
 The class W[$t$], $t\geq 1$ contains problems FPT-reducible to WCS restricted to circuits with weft at most $t$, where the weft of a boolean circuit is the number of gates with unbounded in-degree in any path from the input vertices to the output vertex.
Unless W[$t$]=FPT, any W[$t$]-hard problem is not fixed-parameter tractable. A problem is W[$t$]-hard if any problem in W[$t$] is FPT-reducible to it, and W[$t$]-complete if it is moreover in W[$t$]. 
There is the following hierarchy: FPT $\subseteq$ W[$1$] $\subseteq$ W[$2$] $\subseteq \ldots \subseteq$ XP.
}

We now recall some definitions about approximation complexity. See the book~\cite{ausiello} for details.

Given an optimization problem $A$ in NPO and an instance $I$ of $A$, we denote  by $opt_A(I)$ the optimum value of $I$  (or $opt(I)$ if there is no ambiguity), and by $val(S)$ the value of a feasible solution $S$ of instance $I$. The {\em
performance ratio\/} of $S$ (or {\it approximation factor}) is
$r(I,S)=\max\left\{\frac{val(S)}{opt_A(I)},
\frac{opt_A(I)}{val(S)}\right\}.$ 
For a function $f$, an algorithm  is an {\it
$f(|I|)$-approximation algorithm\/}, if for every instance $I$ of the problem,
it returns a solution $S$ such that $r(I,S) \leq f(|I|).$

The class $\apx$ contains all optimization problems that admit a polynomial-time $c$-approximation algorithm for some fixed constant $c$. A \emph{polynomial-time approximation scheme} (PTAS for short) for an optimization problem is an algorithm that, given any fixed constant $\epsilon>0$, returns in polynomial time (in terms of the instance and for fixed $\epsilon$) a solution that is a factor of $1+\epsilon$ away from the optimum. If the exponent of the polynomial does not depend on $\epsilon$, such an algorithm is called an \emph{efficient PTAS} (EPTAS for short). 


\begin{definition}[$L$-reduction \cite{PY91}] \label{def:l-reduc}
Let $A$ and $B$ be two optimization
problems. Then $A$ is said to be {\it $L$-reducible\/} to $B$ if
there are two constants $\alpha, \beta >0$ and two polynomial time computable functions $f$, $g$ such that:
(i) $f$ maps an instance $I$ of $A$ into an
instance $I'$ of $B$ such that $opt_B(I')\leq \alpha \cdot opt_A(I)$,
(ii) $g$ maps each solution $S'$ of
$I'$ into a solution $S$ of $I$ such that $|val(S)-opt_A(I)|\leq
\beta \cdot|val(S')-opt_B(I')|$.
\end{definition}

$L$-reductions are useful in order to apply the following theorem.

\begin{theorem}[\cite{PY91}]\label{thm:L-reduc}
Let $A$ and $B$ be two optimization problems. If $A$ is $L$-reducible to $B$ and $B$ has a PTAS, then $A$ has a PTAS.
\end{theorem}

\confversion{

\section{Positive approximation results for \MaxNEQCL}\label{sec:pos-approx}

}

\fullversion
{
\section{Parameterized complexity of \NEQCLkl}

In this section we explore the parameterized complexity of decision problem \NEQCLkl.

\subsection{Parameterization by $k$ and $\ell$}\label{sec:nat-param}
}

We start with a greedy polynomial-time procedure that always returns (if it exists), a set $|X'|$ that induces at least $|X'|+1$ equivalence classes.

\begin{lemma}\label{lemm:useBondy}
Let $H=(X,E)$ be a twin-free hypergraph and let $k\leq |X|-1$ be an
integer. One can construct, in time $O(k(|X|+|E|))$, a set $C\subseteq X$
of size~$k$ that produces at least~$\min\{|E|,k+1\}$ neighborhood equivalence classes.
\end{lemma}
\begin{proof}
We produce $C$ in an inductive way. First, let $C_1=\{x_1\}$ for an arbitrary vertex $x_1$ of $X$ for which there exists at least one hyperedge of $E$ with $x_1\notin E$ (if such hyperedge does not exist, then all edges are twin edges; since $H$ is twin-free, $|E|\leq 1$ and we are done). Then, for each $i$ with $2\leq i\leq k$, we build $C_{i}$ from $C_{i-1}$ as follows: select vertex $x_i$ as a vertex in $X\setminus C_{i-1}$ such that $C_{i-1}\cup\{x_i\}$ maximizes the number of equivalence classes.

We claim that either we already have at least $|E|$ equivalence classes, or $C_{i}$ induces at least one more equivalence class than $C_{i-1}$. Assume for a contradiction that we have less than $|E|$ equivalence classes, but $C_{i}$ has the same number of equivalence classes as $C_{i-1}$. Since we have less than $|E|$ classes, there is an equivalence class consisting of at least two edges, say $e_1$ and $e_2$. But then, since $H$ is twin-free, there is a vertex $x$ that belongs to exactly one of $e_1$ and $e_2$. But $C_{i-1}\cup\{x\}$ would have more equivalence classes than $C_{i}$, a contradiction since $C_{i}$ was maximizing the number of equivalence classes.

Hence, setting $C=C_k$ finishes the proof.
\end{proof}

{

By Lemma~\ref{lemm:useBondy}, we may always assume that $\ell\geq k+2$ since otherwise, \NEQCLkl is polynomial-time solvable. Moreover, we can always assume that $\ell\leq 2^k$ (otherwise we have a NO-instance). This shows that both parameters $k$ and $\ell$ are linked to each other, and then \NEQCLkl is in FPT parameterized by $k$ if and only if it is in FPT parameterized by $\ell$ or $k+\ell$. 
However, when studying the existence of a polynomial kernel or sharp algorithm time-complexity lower bounds, one may consider both parameters $k$ and $\ell$. 
 
\begin{theorem}
\NEQCLkl is solvable in time $O((|X|+|E|)^{\min\{k,\ell\}+O(1)})$, and thus \NEQCLkl parameterized by $k$ (or $\ell$) belongs to the class XP.
\end{theorem}
\begin{proof}
First, observe that we can compute, for every
$k$-subset $Y$ of $X$, the set of neighborhood equivalence classes produced by $Y$
(simply compute, for each hyperedge $e$ in $E$, the set $e\cap Y$). This
gives an $O((|X|+|E|)^{k+O(1)})$-time algorithm. Therefore, if $k<\ell$, we are done.
Assume then that $k\geq\ell$. Then, we can apply Lemma~\ref{lemm:useBondy} with $k=\ell-1$, and obtain a solution of size at most $\ell-1\leq k$ creating at least $\ell$ classes, in time $O(k(|X|+|E|))$.
\end{proof}

We refer to the book by Flum and Grohe~\cite{FG06} and to Courcelle~\cite{Courcelle} for the definitions of the logics $\Sigma_1$ and MSOL, respectively.

\begin{proposition}\label{prop:logic}
\NEQCLkl can be expressed by a $\Sigma_1$ (and thus MSOL) formula with a number of variables bounded by a function of $k+\ell$.
\end{proposition}
\begin{proof}
In the spirit of the $\Sigma_1$ formula for \VCDIM given in the book~\cite[Theorem 6.5]{FG06}, we define the following $\Sigma_1$ formula for \NEQCLkl that corresponds to the reformulation of the problem in Observation~\ref{obs:reform}.

\begin{align*}
&\exists x_1,\dots, x_k \in X, 
y_1,\dots,y_{\ell} \in E, 
s_1^2, s_1^3, \dots, s_{\ell-1}^{\ell} \in X,\\
&\bigwedge_{\substack{i,j=0\\i<j}}^{\binom{\ell}{2}}
\left(
\left(
\bigvee_{q=1}^{k} (s_i^j = x_q) 
\right)
\wedge
\left(
(s_i^j \in y_i \wedge s_i^j \notin y_j)
\vee (s_i^j \notin y_i \wedge s_i^j \in y_j)
\right)
\right)
\end{align*}

The variables $x_i$ represent the solution vertices, and each variable $y_i$ is a hyperedge that is representative of a distinct equivalence class. For a pair of distinct hyperedges $y_i$ and $y_j$, the variable $s_i^j$ represents a solution vertex that belongs to exactly one of the hyperedges $y_i$ and $y_j$.~\end{proof}

By viewing \NEQCLkl restricted to neighborhood hypergraphs of graphs as a graph problem, we obtain the following consequence of Proposition~\ref{prop:logic} and Courcelle's theorem~\cite{courcelleCW} (see~\cite{courcelleCW} for a definition of cliquewidth).
Although we only need the result for bounded treewidth graphs in the rest of the article, we give the corollary with the stronger result which holds for cliquewidth.

\begin{corollary}\label{courcelle}
\NEQCLkl can be solved in time $f(w)n$ on neighborhood hypergraphs of graphs of order $n$ and cliquewidth at most $w$ for some function $f$.
\end{corollary}

We now prove a hardness result, by slightly modifying the reduction of~\cite{DEF93}.

\begin{theorem}\label{prop:vchard}
\NEQCLkl parameterized by $k$ belongs to W[1]. Moreover, \VCDIM (and thus \NEQCLkl) is W[1]-hard, even for neighborhood hypergraphs of graphs that are either bipartite, split or co-bipartite.
\end{theorem}
\begin{proof}
The membership in W[1] follows from Proposition~\ref{prop:logic} (see the book~\cite{FG06}).

For the hardness part, we modify the reduction of~\cite{DEF93} for \VCDIM in order to strengthen the result. 
We give an FPT-reduction from \textsc{Clique}, where given a graph $G=(V,E)$ and an integer $k$, the question is to determine if there is a clique of size $k$ in $G$. 
We consider instances of \textsc{Clique} parameterized by $k$, where $|V| \geq k$. Clearly we may assume that $k>3$.
We denote by $[k]$ the set $\{1,\dots,k\}$.
We construct a bipartite graph $G'$ and let $k'=k$. 
We will prove that the neighborhood hypergraph of $G'$ has VC dimension at least $k'$ if and only if $G$ has a clique of size $k$. (For simplicity, here we consider \VCDIM as a graph problem and use graph theory terminology instead of neighborhood hypergraphs, as is done for example in~\cite{bousquet,KKRUW97}.)

Let us describe the construction. We first present the reduction and its proof for bipartite graphs, and later show how to modify it to obtain a split or co-bipartite graph.

We build a bipartite graph $G'=((X,F), E')$ (this is the incidence graph of the hypergraph built in the reduction from~\cite{DEF93}).
We define the first partite set of $G'$ as $X=\{(u,i) : u\in V,i\in [k]\}$. 
The vertices of the second partite set $F$ are defined as follows. 
We let $F = F_2 \cup F_{\neq 2}$ such that $F_2 = \{f_{\{(u,i), (v,j)\}}, uv \in E \text{ and } i,j \in [k]\}$, and $F_{\neq 2} = \{f_L, L\subseteq [k]\text{ and } |L| \neq 2\}$. Each vertex $f_{\{(u,i), (v,j)\}}$ is of degree~$2$ and adjacent to $(u,i)$ and $(v,j)$. Each vertex $f_L$ of $F_{\neq 2}$ is adjacent to the vertices in the set $\{ (u,i) : u \in V, i \in L \}$.

The order of $G'$ is $O(k|V|+k^2|E|+2^k)$. We give a simple example of
the reduction in Figure~\ref{fig:reduction-Clique-VCDIM}.

\begin{figure}[ht!]
\centering
\begin{tikzpicture}[scale=0.9, transform shape]

  \node[draw, shape=circle] at (-1,0) (a) {};
    \path (a)+(0,0.5) node {$a$};
  \node[draw, shape=circle] at (0,0) (b) {};
    \path (b)+(0,0.5) node {$b$};
  \node[draw, shape=circle] at (1,0) (c) {};
    \path (c)+(0,0.5) node {$c$};
  \draw (a)--(b)--(c);
  \node at (0,-1) {Instance of \textsc{Clique}:};
  \node at (0,-1.5) {graph $G=(V,E)$, integer $k=3$.};
  
  \begin{scope}[xshift=5cm]
    \node[draw, shape=circle] at (0,3) (a1) {};
    \path (a1)+(-0.75,0) node {$(a,1)$};
    \node[draw, shape=circle] at (0,2) (a2) {};
    \path (a2)+(-0.75,0) node {$(a,2)$};
    \node[draw, shape=circle] at (0,1) (a3) {};
    \path (a3)+(-0.75,0) node {$(a,3)$};
    \node[draw, shape=circle] at (0,0) (b1) {};
    \path (b1)+(-0.75,0) node {$(b,1)$};
    \node[draw, shape=circle] at (0,-1) (b2) {};
    \path (b2)+(-0.75,0) node {$(b,2)$};
    \node[draw, shape=circle] at (0,-2) (b3) {};
    \path (b3)+(-0.75,0) node {$(b,3)$};
    \node[draw, shape=circle] at (0,-3) (c1) {};
    \path (c1)+(-0.75,0) node {$(c,1)$};
    \node[draw, shape=circle] at (0,-4) (c2) {};
    \path (c2)+(-0.75,0) node {$(c,2)$};
    \node[draw, shape=circle] at (0,-5) (c3) {};
    \path (c3)+(-0.75,0) node {$(c,3)$};

    \node at (0,4) {{$X$}};
    
    \node[draw, shape=circle] at (6,5) (a1b1) {};
    \path (a1b1)+(1.1,0) node {$f_{\{(a,1),(b,1)\}}$};
    \node[draw, shape=circle] at (6,4.5) (a1b2) {};
    \path (a1b2)+(1.1,0) node {$f_{\{(a,1),(b,2)\}}$};
    \node[draw, shape=circle] at (6,4) (a1b3) {};
    \path (a1b3)+(1.1,0) node {$f_{\{(a,1),(b,3)\}}$};
    \node[draw, shape=circle] at (6,3.5) (a2b1) {};
    \path (a2b1)+(1.1,0) node {$f_{\{(a,2),(b,1)\}}$};
    \node[draw, shape=circle] at (6,3) (a2b2) {};
    \path (a2b2)+(1.1,0) node {$f_{\{(a,2),(b,2)\}}$};
    \node[draw, shape=circle] at (6,2.5) (a2b3) {};
    \path (a2b3)+(1.1,0) node {$f_{\{(a,2),(b,3)\}}$};
    \node[draw, shape=circle] at (6,2) (a3b1) {};
    \path (a3b1)+(1.1,0) node {$f_{\{(a,3),(b,1)\}}$};
    \node[draw, shape=circle] at (6,1.5) (a3b2) {};
    \path (a3b2)+(1.1,0) node {$f_{\{(a,3),(b,2)\}}$};
    \node[draw, shape=circle] at (6,1) (a3b3) {};
    \path (a3b3)+(1.1,0) node {$f_{\{(a,3),(b,3)\}}$};
    \node[draw, shape=circle] at (6,0.5) (b1c1) {};
    \path (b1c1)+(1.1,0) node {$f_{\{(b,1),(c,1)\}}$};
    \node[draw, shape=circle] at (6,0) (b1c2) {};
    \path (b1c2)+(1.1,0) node {$f_{\{(b,1),(c,2)\}}$};
    \node[draw, shape=circle] at (6,-0.5) (b1c3) {};
    \path (b1c3)+(1.1,0) node {$f_{\{(b,1),(c,3)\}}$};
    \node[draw, shape=circle] at (6,-1) (b2c1) {};
    \path (b2c1)+(1.1,0) node {$f_{\{(b,2),(c,1)\}}$};
    \node[draw, shape=circle] at (6,-1.5) (b2c2) {};
    \path (b2c2)+(1.1,0) node {$f_{\{(b,2),(c,2)\}}$};
    \node[draw, shape=circle] at (6,-2) (b2c3) {};
    \path (b2c3)+(1.1,0) node {$f_{\{(b,2),(c,3)\}}$};
    \node[draw, shape=circle] at (6,-2.5) (b3c1) {};
    \path (b3c1)+(1.1,0) node {$f_{\{(b,3),(c,1)\}}$};
    \node[draw, shape=circle] at (6,-3) (b3c2) {};
    \path (b3c2)+(1.1,0) node {$f_{\{(b,3),(c,2)\}}$};
    \node[draw, shape=circle] at (6,-3.5) (b3c3) {};
    \path (b3c3)+(1.1,0) node {$f_{\{(b,3),(c,3)\}}$};

    \node (x1) at (8,5.5) {};
    \node (x2) at (8,-4) {};
    \draw[snake={brace},segment amplitude=3mm] (x1) to (x2);
    \node at (8.75,0.75) {{ $F_2$}};

    \node[draw, shape=circle, label=right:$f_{\emptyset}$] at (6,-4) (f0) {};
    \node[draw, shape=circle, label=right:$f_{\{1\}}$] at (6,-4.5) (f1) {};
    \node[draw, shape=circle,label=right:$f_{\{2\}}$] at (6,-5) (f2) {};
    \node[draw, shape=circle,label=right:$f_{\{3\}}$] at (6,-5.5) (f3) {};
    \node[draw, shape=circle,label=right:$f_{\{1,2,3\}}$] at (6,-6) (f123) {};

    \node (x3) at (7.75,-3.7) {};
    \node (x4) at (7.75,-6.35) {};
    \draw[snake={brace},segment amplitude=2mm] (x3) to (x4);
    \node at (8.5,-5.1) {{ $F_{\neq 2}$}};
    
  \node at (3,-7) {Instance of \VCDIM corresponding to $(G,k)$:};
  \node at (3,-7.5) {bipartite graph $G'=((X,F),E')$, integer $k=3$.};

    \draw (a1)--(a1b1)--(b1)
    (a1)--(a1b2)--(b2)
    (a1)--(a1b3)--(b3)
    (a2)--(a2b1)--(b1)
    (a2)--(a2b2)--(b2)
    (a2)--(a2b3)--(b3)
    (a3)--(a3b1)--(b1)
    (a3)--(a3b2)--(b2)
    (a3)--(a3b3)--(b3)
    (b1)--(b1c1)--(c1)
    (b1)--(b1c2)--(c2)
    (b1)--(b1c3)--(c3)
    (b2)--(b2c1)--(c1)
    (b2)--(b2c2)--(c2)
    (b2)--(b2c3)--(c3)
    (b3)--(b3c1)--(c1)
    (b3)--(b3c2)--(c2)
    (b3)--(b3c3)--(c3)
    (f123)--(a1)--(f1)--(b1)--(f123)
    (f123)--(a2)--(f2)--(b2)--(f123)
    (f123)--(a3)--(f3)--(b3)--(f123)
    (f1)--(c1)--(f123) (f2)--(c2)--(f123) (f3)--(c3)--(f123);
    
  \end{scope}

\end{tikzpicture}
\caption{Example of the reduction from the proof of Theorem~\ref{prop:vchard}.}
\label{fig:reduction-Clique-VCDIM}
\end{figure}

If there is a clique $C = \{v_1, \ldots, v_k\}$ in $G$, consider the following set $S \subseteq X$ of size $k$: $S = \{(v_1,1), (v_2,2), \ldots, (v_k, k)\}$. 
Now, for any subset $S' \subseteq S$: 
if $|S'|\neq 2$, then for $L=\{i: (v_i,i) \in S'\}$ we have $X \cap f_L = S'$, and the vertex $f_L$ represents the equivalence class $S'$.
If $|S'|=2$ (assume $S' = \{ (v_i,i), (v_j, j) \}$ with $i \neq j$), then $v_iv_j \in E$ because $C$ is a clique. 
Hence, $f_{S'} \in F_2$ and is a representative of the equivalence class $S'$. Thus, $S$ induces all possible $2^k$ equivalence classes.

For the converse direction, let $S$ be a set of $k$ vertices of $G'$ that induces $2^k$ equivalence classes in $G'$. First of all, we claim that the set $S$ cannot contain vertices of both $X$ and $F$. Assume to the contrary that there is a vertex $x\in X$ and a vertex $f\in F$ in $S$. Since $k>3$, we would need a vertex in $G'$ at distance at most~$1$ from both $x$ and $f$ and distinct from $x$ and $f$, which is impossible since $G'$ is bipartite.

If $S \subseteq F$, then we claim that $S$ cannot induce enough equivalence classes. Indeed, $S$ cannot contain vertices of $F_2$ since each vertex $f$ in $F_2$ has degree~$2$ and each vertex $s$ in $S$ should have degree at least $2^{k-1}-1>2$ (because half of the possible equivalence classes contain $s$). 
Thus, if $S \subseteq F$ then $S \subseteq F_{\neq 2}$. 
But for all $i \in [k]$ and for all $u,v\in V$, $(u,i)$ and $(v,i)$ have the same set of neighbors in $F_{\neq 2}$. 
Thus, $S$ can induce at most $2k$ classes: at most $k$ with a representative in $X$, and at most $k$ with a representative in $F$ (the latter ones are classes of size~$1$). This is a contradiction since $k>3$ and $S$ is a set that induces $2^k$ equivalence classes.

Therefore, $S \subseteq X$. It remains to show that the set $S_G=\{v_i\in V: (v_i,j)\in S\}$ forms a clique of size $k$ in $G$. For each subset $S'$ of $S$ with size at least~$3$, the corresponding equivalence class can only be realized by a vertex $f$ with $N(f)\cap S=S'$. Thus, $f$ belongs to $F_{\neq 2}$ and $f$ has at least three neighbours in $S$. But the number of vertices of $F_{\neq 2}$ with at least three neighbours in $X$ is precisely the number of subsets of $S$ of size at least~$3$ (because $|S|=k$ and there are exactly $2^k-{k\choose 2}$ vertices in $F_{\neq 2}$, one for each subset of $[k]$ whose size is not~$2$). Therefore, each vertex in $F_{\neq 2}$ with at least three neighbours in $X$ is the (unique) representative for an equivalence class corresponding to a subset $S'$ of $S$ with $|S'|\geq 3$. Now, consider the subsets $S' \subseteq S$ with $|S'| = 2$ (assume that $S'=\{(u,i),(v,j)\}$). There must be some $f$ in $F$ with $N(f) \cap S = S'$. By the previous argument and since $f$ needs at least two neighbours in $X$, we have $f\notin F_{\neq 2}$, and thus $f\in F_2$.
Hence, by the definition of $F_2$, $u$ and $v$ are distinct and adjacent in $G$. Therefore, $S_G$ is a clique of size $k$ in $G$. This completes the proof for bipartite graphs.

We may modify this construction to obtain a split graph (then we make $X$ to a clique) or a co-bipartite graph (then we make both $X$ and $F$ into cliques). The arguments are almost the same. The main difference is in the converse direction, when proving that no solution $S$ can contain vertices of both $X$ and $F$. If $G'$ is co-bipartite, this would make it impossible to induce the empty equivalence class, a contradiction. If $G'$ is split, assume to the contrary that we have $x\in X$ and $f\in F$ and both $x,f\in S$. Then, we have $2^{k-2}$ equivalence classes that contain $f$ but do not contain $x$. The vertices corresponding to these classes necessarily belong to $F$ (since $X$ is a clique), but since $F$ is an independent set there is at most one such vertex ($f$ itself). Since $k>3$ we have $2^{k-2}>2$, a contradiction.

The case that $S\subseteq F$ is handled similarly as before, by observing that $F$ or $X$ being a clique would not help to increase the number of equivalence classes induced by $S$. The rest of the proof is the same.
\end{proof}

We now apply the technique of \emph{random separation} introduced in~\cite{randomsep}. Let $d_G(u,v)$ be the distance between $u$ and $v$ in $G$, and for two subsets of vertices of $G$, $V_1$ and $V_2$, let $d_G(V_1,V_2) = \min\{d_G(u,v) | u\in V_1, v \in V_2\}$.

\begin{theorem}[{\cite[Theorem 4]{randomsep}}]\label{thm:randomsep}
Let $G=(V,E)$ be a graph of maximum degree~$\Delta$. Let $\phi:2^V\to \mathbb{R}\cup\{-\infty,+\infty\}$ be an objective function to be optimized. Then, it takes $O(f(k,\Delta)|V|^{\max\{c',c+1\}}\log |V|)$ time (for some computable function $f$) to find a set $V'$ of $k$ vertices in $G$ that optimizes $\phi(V')$ if the following conditions are satisfied.
\begin{enumerate}
\item For each subset $V'$ of $V$ of size at most $k$, $\phi(V')$ can be computed in time $O(g(k,\Delta)|V|^c)$ for some computable function $g$ and constant $c>0$.
\item There is a positive integer $i$ computable in time $O(h(k,\Delta)|V|^{c'})$ for some computable function $h$ and constant $c'>0$ such that for each pair $V_1,V_2$ of subsets of $V$ with $|V_1|+|V_2|\leq k$, if $d_G(V_1,V_2)>i$ 
then $\phi(V_1\cup V_2)=\phi(V_1)+\phi(V_2)$.
\end{enumerate}
\end{theorem}

\begin{theorem}\label{thm:boundeddegree}
\NEQCLkl is solvable in time $2^{O(k\Delta^2)}n\log n$ on neighborhood hypergraphs of graphs of maximum degree~$\Delta$ and order $n$.
\end{theorem}
\begin{proof}
We apply Theorem~\ref{thm:randomsep}. For the first condition, one can compute in time $O(k\Delta)$ the number of equivalence classes induced by any set of $k$ vertices, so $c=0$ and $g(k,\Delta)=k\Delta$.

The second condition is satisfied for $i=2$ by the definition of an equivalence class. Indeed, assume we have two sets $V_1$ and $V_2$ at distance at least~$3$ apart. Then, no vertex is adjacent both to a vertex of $V_1$ and $V_2$, and thus, each equivalence class induced by $V_1\cup V_2$ is induced by $V_1$ alone or $V_2$ alone. Thus $h(d,\Delta)=1$ and $c'=0$.

For the running time, although it is not stated explicitly in the meta-theorem Theorem~\ref{thm:randomsep}, it can be deduced by a careful analysis of the proof sketch in~\cite{randomsep}. Indeed, the main feature of the technique is to generate some suitable $2$-partitions of $V(G)$, that are called \emph{$i$-separating}. This step takes time $2^{O(k\Delta^i)}$. Then, for each such partition, the algorithm runs in time $O((g(k,\Delta)+k+\Delta)n^{\max\{c',c+1\}}\log n)$. There is an additional preprocessig phase running in time $h(k,\Delta)n$. Thus, this amounts in a total running time of $O((h(k,\Delta)+2^{O(k\Delta^i)}+g(k,\Delta)+k+\Delta)n^{\max\{c',c+1\}}\log n)$, which in our case is $2^{O(k\Delta^2)}n\log n$. We refer the interested reader to~\cite{randomsep} for all details.
\end{proof}


\subsection{Dual parameterizations}\label{sec:dual-param}

We now study the four dual parameterizations for \NEQCLkl, that is, parameters $|E|-k$,  $|X|-k$, $|E|-\ell$ and $|X|-\ell$. Note that \PBDS is NP-hard for neighborhood hypergraphs of graphs (see for example~\cite{F13}). This corresponds to \NEQCLkl with $|X|=|E|=\ell$ and hence \NEQCLkl is not in XP for parameters $|E|-\ell$ and $|X|-\ell$.

\begin{theorem}\label{thm:dual-params}
\NEQCLkl belongs to XP but is W[1]-hard when parameterized by $|X|-k$.
\end{theorem}
\begin{proof}
We first show that  \NEQCLkl is in XP parameterized by  $|X|-k$. Indeed, given $H=(X,E)$, it suffices to check, for each of the $\binom{|X|}{k}=\binom{|X|}{|X|-k}=O(|X|^{|X|-k})$ subsets of size $k$ of $X$, whether it induces at least $\ell$ equivalence classes in $H$ (the latter can be done in polynomial time by computing all intersections of the given $k$-set with the edges of $H$ and counting how many distinct ones are obtained).

For the hardness part, we give a slight modification of a reduction from~\cite{CGJSY12}\footnote{The reduction in~\cite{CGJSY12} is for \PBTC parameterized by $|E|-k$, which is equivalent by taking the dual hypergraph.}.
We reduce from \textsc{Independent Set} parameterized by the solution size $s$ to \PBDS parameterized by $|X|-k$.

Given a graph $G$ on $n_G$ vertices and $m_G$ edges, we construct a hypergraph $H(G)=(X,E)$ as follows. The vertex set $X$ contains $n_G+m_G$ vertices: for each vertex $v$ of $G$, we have the vertex $x_v$ in $X$; moreover, for each edge $e$ of $G$, we have the vertex $x_e$ in $X$. The hyperedge set $E$ contains $2m_G+1$ hyperedges: for each edge $e$ of $G$, we have the hyperedges $E_e$ and $E'_e$ in $E$, where $E_e=\{x_v, v\in e\}\cup\{x_e\}$ and $E'_e=\{x_e\}$. Moreover, we have an additional empty hyperedge in $E$. Now, we claim that $G$ has an independent set of size $s$ if and only if in $H(G)$, one can find a set of $k=|X|-s$ vertices that induces $\ell=|E|=2m_G+1$ equivalence classes.

Suppose first that $G$ has an independent set $I$ of size $s$. Then, the vertex set $X\setminus\{x_v, v\in I\}$ induces $|E|$ equivalence classes in $H(G)$.

Conversely, if $S$ is a solution to \PBDS of size $k=|X|-s$, then $S$ must contain each vertex $x_e$ of $H(G)$ (otherwise the hyperedge $\{x_e\}$ would not be distinguished from the empty hyperedge). Moreover, we claim that the set of vertices of $G$ $\{v: x_v\notin S\}$ is an independent set (of size $s$) in $G$. Assume by contradiction that there is an edge $e$ between two vertices $u$ and $v$ in $G$ and $x_u,x_v\notin S$. Then, $E_e$ and $E'_e$ are in the same equivalence class, a contradiction.

Since $|X|-k=s$, this reduction shows that \PBDS and \NEQCLkl parameterized by $|X|-k$ are W[1]-hard.
\end{proof}

We do not know whether \NEQCLkl parameterized by $|X|-k$ belongs to W[1]. It is proved in~\cite{CGJSY12} that this holds for the special case of $\ell=|E|$, that is, of \PBDS, by a reduction to \textsc{Set Cover} parameterized by $|E|-k$.

Regarding \VCDIM parameterized by $|X|-k$ (a special case of \NEQCLkl parameterized by $|X|-k$), it is in XP by Theorem~\ref{thm:dual-params}. However, it is not clear whether it is in FPT or W[1]-hard.



Note that if $|X|\leq |E|$, we have $|X|-k\leq |E|-k$ and then \NEQCLkl parameterized by $|E|-k$ is in XP by Theorem~\ref{thm:dual-params}. Nevertheless, we cannot use this in the general case, and need the following dedicated discussion.

\begin{proposition}
\NEQCLkl parameterized by $|E|-k$ belongs to XP.
\end{proposition}
\begin{proof}
Since $H$ is twin-free, we have $|X|\leq 2^{|E|}$. Thus, if $|E|\leq 2(|E|-k)$, the whole instance is a kernel and in fact any brute-force algorithm is an FPT algorithm. Hence we also assume that $|E|-k<|E|/2$.

Now, we first guess the set of $\ell$ hyperedges representing the $\ell$ equivalence classes in case we have a YES-instance. This takes time $\binom{|E|}{\ell} = \binom{|E|}{|E|-\ell} \leq \binom{|E|}{|E|-k}$ (since $k \leq \ell$ and $|E|-k<|E|/2$). 
Once the set of $\ell$ hyperedges is fixed, we remove the other hyperedges to obtain a new sub-hypergraph $(X,E')$. Now, finding $k$ vertices that induce exactly $\ell$ classes is precisely \PBDS on $(X,E')$, which is known to be in FPT for parameter $|E'|-k$~\cite{CGJSY12} (more precisely, authors of~\cite{CGJSY12} give this result for \PBTC for parameter $|X|-k$ which is equivalent to \PBDS for parameter $|E'|-k$ by taking the dual hypergraph). Since $|E'|-k=\ell-k\leq |E|-k$ this is an XP algorithm.
\end{proof}

We do not know if \NEQCLkl is in FPT or W[1]-hard for parameter $|E|-k$.
The case $\ell=|E|$ (that is, \PBDS) is shown to be in FPT~\cite{CGJSY12}. Moreover, \VCDIM is (trivially) in FPT when parameterized by $|E|-k$.
Indeed, any instance of \VCDIM with $k>\log_2 |E|$ is trivially a NO-instance. 
On the other hand, when $k\leq \log_2 |E|$, then $|E|-k\geq |E|-\log_2 |E|$.
Therefore, a trivial brute-force algorithm for \VCDIM is a fixed parameter tractable algorithm.

}


\section{Approximation complexity of \MaxNEQCL}

We now study the complexity of approximating \MaxNEQCL.

\subsection{Positive results for \MaxNEQCL}\label{sec:pos-approx}

\begin{proposition}\label{prop:2kk}
\MaxNEQCL is $\frac{\min\{2^k,|E|\}}{k+1}$-approximable in polynomial time. For hypergraphs with VC dimension at most $d$, \MaxNEQCL is $k^{d-1}$-approximable. For hypergraphs with maximum degree~$\Delta$, \MaxNEQCL is $\frac{\Delta+1}{2}$-approximable. 
\end{proposition}
\begin{proof}
By Lemma~\ref{lemm:useBondy}, we can always compute in polynomial time, a solution with at least $k+1$ neighborhood equivalence classes (if it exists; otherwise, we solve the problem exactly). 
Since there are at most $\min\{2^k,|E|\}$ possible classes, the first part of the statement follows. Similarly, by Corollary~\ref{cor:VCdim}, if the hypergraph has VC dimension at most~$d$, there are at most $k^{d}+1$ equivalence classes, and $\frac{k^{d}+1}{k+1}\leq k^{d-1}$. Finally, if the maximum degree is at most $\Delta$, by Theorem~\ref{thm:delta-LB} there are at most $\frac{k(\Delta+1)+2}{2}$ possible classes (and when $\Delta\geq 1$, $\frac{k(\Delta+1)+2}{2(k+1)}\leq \frac{\Delta+1}{2}$).
\end{proof}

\begin{corollary}\label{cor:approx-VCdim}
For hypergraphs of VC dimension at most $d$, \MaxNEQCL is $|E|^{(d-1)/d}$-approximable.
\end{corollary}
\begin{proof}
By Proposition~\ref{prop:2kk}, we have a $\min\{k^{d-1},|E|/k\})$-approximation. If $k^{d-1} < |E|^{(d-1)/d}$ we are done. Otherwise, we have $k^{d-1}\geq |E|^{(d-1)/d}$ and hence $k\geq |E|^{1/d}$, which implies that $\frac{|E|}{k}\leq |E|^{(d-1)/d}$.
\end{proof}

For examples of concrete applications of Corollary~\ref{cor:approx-VCdim}, hypergraphs with no $4$-cycles in their bipartite incidence graph\footnote{In the dual hypergraph, this corresponds to the property that each pair of hyperedges have at most one common element, see for example~\cite{KAR00}.} have VC dimension at most~$3$ and hence we have an $|E|^{2/3}$-approximation for this class.
Hypergraphs with maximum edge-size $d$ also have VC dimension at most $d$. 
Other examples, arising from graphs, are neighborhood hypergraphs of: $K_{d+1}$-minor-free graphs (that have VC dimension at most $d$~\cite{BT15}); graphs of rankwidth at most~$r$ (VC dimension at most $2^{2^{O(r)}}$~\cite{BT15}); interval graphs (VC dimension at most~$2$~\cite{BLLPT15}); permutation graphs (VC dimension at most~$3$~\cite{BLLPT15}); line graphs (VC dimension at most~$3$); unit disk graphs (VC dimension at most~$3$)~\cite{BLLPT15}; $C_4$-free graphs (VC dimension at most~$2$); chordal bipartite graphs (VC dimension at most~$3$~\cite{BLLPT15}); undirected path graphs (VC dimension at most~$3$~\cite{BLLPT15}). 
Typical hypergraph classes with unbounded VC dimension are neighborhood hypergraphs of bipartite graphs, co-bipartite graphs, or split graphs.

In the case of hypergraphs with no $4$-cycles in their bipartite incidence graph (for which \MaxNEQCL has an $|E|^{2/3}$-approximation algorithm by Corollary~\ref{cor:approx-VCdim}), we can also relate \MaxNEQCL to \textsc{Max Partial Double Hitting Set}, defined as follows.

\problemopt{\textsc{Max Partial Double Hitting Set}}{A hypergraph $H=(X,E)$, an integer $k$.}{A subset $C \subseteq X$ of size $k$ maximizing the number of hyperedges containing at least two elements of $C$.}

\begin{theorem}\label{thm:DoubleSetCover}
Any $\alpha$-approximation algorithm for \textsc{Max Partial Double Hitting Set} can be used to obtain a $4\alpha$-approximation algorithm for \MaxNEQCL on hypergraphs without $4$-cycles in their bipartite incidence graph.
\end{theorem}
\begin{proof} Let $H=(X,E)$ be a hypergraph without $4$-cycles in its bipartite incidence graph, and let $C\subseteq X$ be a subset of vertices. 
Since $H$ has no $4$-cycles in its bipartite incidence graph, note that if some hyperedge contains two vertices of $X$, then no other hyperedge contains these two vertices.
Therefore, the number of equivalence classes induced by $C$ is equal to the number of hyperedges containing at least two elements of $C$, plus the number of equivalence classes corresponding to a single (or no) element of $C$. 
Therefore, the maximum number $opt(H)$ of equivalence classes for a set of size $k$ is at most $opt_{2HS}(H)+k+1$, where $opt_{2HS}(H)$ is the value of an optimal solution for \textsc{Max Partial Double Hitting Set} on $H$. 
Observing that $opt_{2HS}(H)\geq\frac{k}{2}$ (since one may always iteratively select pairs of vertices covering a same hyperedge to obtain a valid double hitting set of $H$), we get that $opt(H)\leq 3opt_{2HS}(H) +1  \leq 4opt_{2HS}(H)$. 
Moreover, in polynomial time we can apply the approximation algorithm of \textsc{Max Partial Double Hitting Set} to $H$ to obtain a set $C$ inducing at least $\frac{opt_{2HS}(H)}{\alpha}$ neighborhood equivalence classes. 
Thus, $C$ induces at least $\frac{opt(H)}{4\alpha}$ neighborhood equivalence classes.
\end{proof}

The approximation complexity of \textsc{Max Partial Double Hitting Set} is rather open. The problem admits no PTAS (unless NP$\subseteq \cap_{\varepsilon >0}$ BPTIME($2^{n^{\varepsilon}}$))~\cite{K04}, but there is a simple $O(|E|^{1/2})$-factor approximation algorithm (even in the general case) based on the greedy constant-factor algorithm for \textsc{Max Partial Hitting Set}, see~\cite{maxdoubleHSapprox}. We deduce the following corollary, which improves on the $O(|E|^{2/3})$-factor approximation algorithm of Corollary~\ref{cor:approx-VCdim}.

\begin{corollary}\label{cor:approx-maxdoubleHS}
\MaxNEQCL can be approximated within a factor of $O(|E|^{1/2})$ on hypergraphs without $4$-cycles in their bipartite incidence graph.
\end{corollary}

Another special case corresponds to problem \textsc{Max Densest Subgraph} (which, given an input graph, consists of maximizing the number of edges of a subgraph of order~$k$). This is precisely \textsc{Max Partial Double Hitting Set} restricted to hypergraphs where each hyperedge has size at most~$2$ (that is, to graphs). Such instances can be assumed to contain no $4$-cycles in their bipartite incidence graph (a $4$-cycle would imply the existence of two twin hyperedges). 
The best known approximation ratio for \textsc{Max Densest Subgraph} is $O(|E|^{1/4})$~\cite{BCCFV10}.\footnote{Formally, it is stated in~\cite{BCCFV10} as an $O(|V|^{1/4})$-approximation algorithm, but we may assume that the input graph is connected, and hence $|V|=O(|E|)$.} We deduce from this result, the following corollary of Theorem~\ref{thm:DoubleSetCover} for hypergraphs of hyperedge-size bounded by~$2$. This improves on both $O(|E|^{1/2})$-approximation algorithms given by Corollaries~\ref{cor:approx-VCdim} and~\ref{cor:approx-maxdoubleHS} for this case.

\begin{corollary}
Any $\alpha$-approximation algorithm for \textsc{Max Densest Subgraph} can be used to obtain a $4\alpha$-approximation algorithm for \MaxNEQCL on hypergraphs with hyperedges of size at most~$2$. 
In particular, there is a polynomial-time $O(|E|^{1/4})$-approximation algorithm for this case.
\end{corollary}

We will now apply the following result from~\cite{BCNS14}.

\begin{lemma}[\cite{BCNS14}]\label{lemma:FPT-approx}
If an optimization problem is $r_1(k)$-approximable in FPT time with respect to parameter $k$ for \emph{some} strictly increasing function $r_1$ depending solely on $k$, then it is also $r_2(n)$-approximable in FPT time with respect to parameter $k$ for \emph{any} strictly increasing function $r_2$ depending solely on the instance size $n$.
\end{lemma}

Using Proposition~\ref{prop:2kk} showing that \MaxNEQCL is $\frac{2^k}{k+1}$-approximable and Lemma~\ref{lemma:FPT-approx}, we directly obtain the following.

\begin{corollary}
For any strictly increasing function $r$, \MaxNEQCL parameterized by $k$ is $r(n)$-approximable in FPT-time.
\end{corollary}

In the following, we establish PTAS algorithms for \MinPBDS and \MaxNEQCL on neighborhood hypergraphs of planar graphs using the layer decomposition technique introduced by Baker~\cite{Bak94}. (Note that the problem is not closed under edge contractions, hence we cannot use the classic bidimensionality framework~\cite{bidim} to obtain this result. For example, consider the graph $P_3$, which is a YES-instance for $k=1$ and $\ell=2$, but after contracting an edge it becomes $P_2$, which is a NO-instance.)

Recall that the two problems are NP-hard on planar graphs: see~\cite{F13}.

Given a planar embedding of an input graph, we call the
vertices which are on the external face
\emph{level 1 vertices}. By induction, we define \emph{level $t$ vertices}
as the set of vertices which are on the external face after removing
the vertices of levels smaller than $t$~\cite{Bak94}. A
planar embedding is \emph{$t$-level} if it has no vertices of level
greater than $t$. If a planar graph is $t$-level, it has a
$t$-outerplanar embedding.

\begin{theorem}
\label{maxPTAS}
\MaxNEQCL on neighborhood hypergraphs of planar graphs admits an EPTAS.
\end{theorem}
\begin{proof}
Let $G$ be a planar graph with a $t$-level planar embedding for some integer $t$.
We aim to achieve an approximation ratio of $1+\varepsilon$. Let $\lambda= 2+\lceil\frac{3}{\varepsilon} \rceil$. 

Let $G_{i}$ ($0\leq i\leq \lambda$) be the graph obtained from $G$ by removing the vertices on levels $i\bmod (\lambda+1)$. 
Thus, graph $G_i$ is the disjoint union of a set of subgraphs $G_{ij}$ ($0\leq j\leq p$ with $p=\lceil \frac{t+i}{\lambda+1}\rceil$) where $G_{i0}$ is induced by the vertices on levels $0,\ldots,i-1$ (note that $G_{00}$ is empty) and $G_{ij}$ with $j\geq 1$ is induced by the vertices on levels $(j-1)(\lambda+1)+i+1, \ldots, j(\lambda+1)+i-1$.

In other words, each subgraph $G_{ij}$ is the union of at most $\lambda$ consecutive levels and is thus $\lambda$-outerplanar, and it forms a connected component of $G_i$ if $G$ is connected. Hence, $G_{i}$ also is $\lambda$-outerplanar and has treewidth at most $3\lambda-1$~\cite{Bod98}. Moreover a tree-decomposition of optimal treewidth (at most $3\lambda-1$) can be computed in polynomial time given the $\lambda$-outerplanar embedding \cite{Bod96}. By Corollary~\ref{courcelle}, for any integer $t$ and any subgraph $G_{ij}$, we can determine an optimal set $S_{ij}^t$ of $t$ vertices of $G_{ij}$ that maximizes the number of (nonempty) induced equivalence classes in $G_{ij}$ in time $O(f(\lambda)|V(G_{ij})|)$ (recall that any graph of bounded treewidth also has bounded cliquewidth). 

We then use dynamic programming to construct an optimal solution for $G_i$ by choosing the distribution of the $k$ solution vertices among the connected components of $G_i$. 

Denote by $S_i(q, y)$ a solution corresponding to the maximum feasible number of equivalence classes induced by a set of $y$ vertices of $G_i$ ($0\leq  y\leq k$) among the union of the first $q$ subgraphs $G_{i1}, \ldots, G_{iq}$ ($1\leq  q\leq p$). $S_i(1,y)$ is the optimal solution $S_{i1}^y$ for $G_{i1}$, as computed in the previous paragraph.
We then compute $S_i(q,y)$ as the union of $S_{iq}^{x^*}$ and $S_i(q-1,y-x^*)$, where $val(S_{iq}^{x^*})+val(S_i(q-1,y-x^*))=\max_{0\leq x\leq y} (val(S_{iq}^x) + val(S_i(q-1,y-x)))$.


 
Let $S_i=S_i(p,k)$. Among $S_0,\ldots, S_{\lambda}$, we choose the best solution, that we denote by $S$. 
We now prove that $S$ is an $(1+\varepsilon)$-approximation of the optimal value $opt(G)$ for \MaxNEQCL on $G$. 
Let $S_{opt}$ be an optimal solution of $G$. Then, we claim that there is at least one integer $r$ such that at most $3/(\lambda+1)$ of the equivalent classes induced by $S_{opt}$ in $G$ are lost when we remove vertices on the levels congruent to $r\bmod (\lambda + 1)$. Indeed, denote by $\ell_i$ the number of equivalence classes induced by $S_{opt}$ and that contain a vertex from a level congruent to $i\bmod (\lambda+1)$. Since an equivalence class intersects vertices of at most three levels, we have that each class is counted at most three times in the sum $\sum_{0\leq i\leq \lambda}\ell_i$, and we have $\sum_{0\leq i\leq \lambda}\ell_i\leq 3val(S_{opt})$. By an averaging argument we obtain the existence of claimed $r$.

Thus, $val(S)\geq val(S_r)\geq opt(G)-\frac{3opt(G)}{\lambda+1}=\frac{\lambda-2}{\lambda+1}opt(G)\geq \frac {opt(G)}{1+\varepsilon}$, which completes the proof.

The overall running time of the algorithm is $\lambda$ times the running time for graphs of treewidth at most $3\lambda-1$, that is, $O(g(\lambda)n)$ for some function $g$.
\end{proof}

As a side result, using the same technique, we provide the following theorem about \MinPBDS, which is an improvement over the $7$-approximation algorithm that follows from~\cite{RS84} (in which it is proved that any YES-instance satisfies $\ell\leq 7k$) and solves an open problem from~\cite{F13} (indeed this is equivalent to \textsc{Min Identifying Code} on planar graphs).

\problemopt{\MinPBDS}{A hypergraph $H=(X,E)$, and an integer $k$.}{A set $C\subseteq X$ of minimum size  that induces $|E|$ distinct equivalence classes?}

\confversion{Due to space constraints, its proof is omitted.}

\begin{theorem}\label{minPTAS}
\MinPBDS on neighborhood hypergraphs of planar graphs admits an EPTAS. 
\end{theorem}
{
\begin{proof} Let $G$ be a planar graph with a $t$-level planar embedding for some integer $t$.
We aim to achieve an approximation ratio of $1+\varepsilon$. 
Let $\lambda= \lceil\frac{2}{\varepsilon} \rceil$. 
Let $G_{ij}$ be the subgraph of $G$ induced by the vertices on levels $j\lambda+i, \ldots,\max\{t,(j+1)\lambda+i+1\}$ ($0\leq i\leq \lambda-1$, $0\leq j\leq \frac{t-i-1}{\lambda})$. 
In other words, each subgraph $G_{ij}$ is the union of at most $\lambda+2$ consecutive levels of the embedding, and the subgraphs $G_{ij}$ and $G_{i(j+1)}$ have two levels in common. 
Each subgraph $G_{ij}$ is thus $(\lambda+2)$-outerplanar and has treewidth at most $3\lambda+5$~\cite{Bod98}. 
Thus, in $G_{ij}$, we can efficiently determine   in time $O(f(\lambda)|V(G_{ij})|)$   a minimum-size set $S_{ij}$ such that each vertex belongs to a distinct neighborhood equivalence class, and has at least one vertex from its closed neighborhood in $S_{ij}$ (this can be done using Courcelle's theorem~\cite{Courcelle}, as noted in~\cite[Chapter 2]{Moncel}). 
Let $S_i= \bigcup_j S_{ij}$. We claim that $S_i$ is a distinguishing set of $G$. Indeed, consider two vertices $x$ and $y$ of $G$. 
If they belong to a same subgraph $G_{ij}$, they are clearly separated by some solution vertex in $S_{ij}$. Now, assume that $x$ and $y$ do not both belong to the same subgraph $G_{ij}$. 
In particular, they must belong to two levels of the embedding that are at distance at least~$3$ apart. Thus, they are not adjacent and do not have any common neighbor. Recall that each vertex of $G_{ij}$ has a vertex of $S_{ij}$ in its closed neighborhood. 
Thus, this vertex separates $x$ and $y$.

Now, among $S_0,\ldots, S_{\lambda-1}$, we choose the minimum-size solution, that we denote by $S$. 
We are going to prove that $S$ is an $(1+\varepsilon)$-approximation. 
Let $S_{opt}$ be an optimal solution for $G$, and let $opt_i$ be the number of vertices that are in $S_{opt}$ and on levels $i\bmod \lambda$ and $(i+1)\bmod \lambda$. 
 We can easily see that there is at least one $r$ ($0\leq r\leq \lambda-1$) such that $opt_r\leq \frac{2|S_{opt}|}{\lambda}\leq \varepsilon |S_{opt}|$, since $\sum_{i=0}^{\lambda-1} opt_i=2|S_{opt}|$. 
 Thus $|S|\leq |S_r| \leq |\cup_j S_{rj}|\leq \sum_j |S_{rj}|\leq \sum_j |S_{opt}\cap V(G_{rj})|=|S_{opt}| + opt_r=|S_{opt}|+opt_r\leq (1+\varepsilon)|S_{opt}|$.

The overall running time of the algorithm is $\lambda$ times the running time for graphs of treewidth at most $3\lambda+5$, that is, $O(g(\lambda)n)$ for some function $g$.
 \end{proof}
}

\fullversion{
\subsection{Hardness results for \MaxNEQCL}\label{sec:neg-approx}
}

\confversion{
\section{Hardness of approximation results for \MaxNEQCL}\label{sec:neg-approx}
}

We define \MaxVCDIM as the maximization version of \VCDIM.

\problemopt{\MaxVCDIM}{A hypergraph $H=(X,E)$.}{A maximum-size shattered subset $C \subseteq X$ of vertices.}

The complexity of \MaxVCDIM is not well understood (it is mentioned as an outstanding open problem in~\cite{CJK02}), but a few results have been obtained. The problem is trivially $\log_2 |E|$-approximable (by an algorithm returning a single vertex). A lower bound on the running time of a potential PTAS has been proved (under the assumption that SNP problems cannot be solved in subexponential time)~\cite{CHKX06}. Very recently, a breakthrough about this question has been obtained under the \emph{randomized Exponential Time Hypothesis} (rETH) assumption (stating that \textsc{$3$-Sat} does not admit any subexponential Monte-Carlo algorithm)~\cite{inapproxVCDIM}. In the following, we establish a connection between the approximability of \MaxVCDIM and \MaxNEQCL.

In the following, we establish a connection between the approximability of \MaxVCDIM and \MaxNEQCL. 

\begin{theorem}\label{thm:VCDIM-inapprox}
Any $2$-approximation algorithm for \MaxNEQCL can be transformed into
a randomized approximation algorithm for \MaxVCDIM with (expected) polynomial
overhead in the running time and with approximation factor at most $2-\tfrac{4}{2+d}=2-o(1)$ (with $d$ the VC dimension of the input hypergraph).
\end{theorem}
\begin{proof}
Let $H$ be a hypergraph on $n$ vertices that is an instance for \MaxVCDIM, and suppose we have a $c$-approximation algorithm $\mathscr{A}$ for \MaxNEQCL.

We run $\mathscr{A}$ with $k=1,\ldots,\log_2 |X|$, and let $k_0$ be the
largest value of $k$ such that the algorithm outputs a solution with
at least $\frac{2^k}{c}$ neighborhood equivalence classes. Since $\mathscr{A}$ is a
$c$-approximation algorithm, we know that the optimum for \MaxNEQCL
for any $k>k_0$ is less than $2^k$. This implies that the
VC dimension of $S$ is at most $k_0$.

Now, let $X$ be the solution set of size $k_0$ computed by
$\mathscr{A}$, and let $H_X$ be the sub-hypergraph of $H$ induced by
$X$. By our assumption, this hypergraph has at least
$\frac{2^{k_0}}{c}$ distinct edges. We can now apply the Sauer-Shelah Lemma
(Theorem~\ref{thm:sauer}).

We have $c=2$, and we apply the lemma with $|X|=k_0$ and $d=\frac{k_0}{2}+1$; it
follows that the VC dimension of $H_X$ (and hence, of $H$) is at
least~$\frac{k_0}{2}+1$.  By the constructive proof of
Theorem~\ref{thm:sauer}, a shattered set $Y$ of this size can be
computed via a randomized algorithm in expected polynomial time~\cite{A98,M01}. Since by the previous paragraph, the
VC dimension of $H$ is at most~$k_0$, set $Y$ is an $f$-approximation for \MaxVCDIM on $H$, where $f\leq\frac{k_0}{k_0/2+1}=2-\frac{4}{k_0+2}\leq 2-\frac{4}{d+2}$, with $d$ the VC dimension of $H$.
\end{proof}

We note that the previous proof does not seem to apply for any other
constant than~$2$, because the Sauer-Shelah Lemma would not apply. 




\medskip

We shall now prove that \MaxNEQCL has no PTAS (unless P$=$NP), even for neighborhood hypergraphs of graphs of bounded maximum degree.
Recall that in this setting, \MaxNEQCL is in $\apx$ by Proposition~\ref{prop:2kk}.

Before proving this negative result, we need an intermediate negative result for \MaxVC (also known as \textsc{Max $k$-Vertex Cover}~\cite{C07}), which is defined as follows.

\problemopt{\MaxVC}{A graph $G=(V,E)$, an integer $k$.}{A subset $S \subseteq V$ of size $k$ covering the maximum number of edges.}

\fullversion{\MaxVC has no PTAS (unless P$=$NP), even when $G$ is cubic~\cite{P94}. Since the proof of~\cite{P94} is not formulated in this language, we give a proof for completeness.}

\begin{proposition}[\cite{P94}]\label{prop:MaxVC-APX}
\MaxVC has no PTAS (unless P$=$NP), even for cubic graphs.
\end{proposition}
{
\begin{proof}
Suppose there is a PTAS algorithm $\mathscr{A}$ for \MaxVC, thus giving a $(1+\varepsilon)$-approximate solution on a graph $G=(V,E)$, $|V|=n, |E|=m=3n/2$.
We will use $\mathscr{A}$ to solve \MinVC by first invoking it for $k=1,2,\dots, n$. 
We claim that for some $k=k_0$, $\mathscr{A}$ produces a solution where $k_0$ vertices cover at least $\frac{m}{1+\varepsilon}$ edges. 
Indeed, when $k=opt_{VC}(G)$ (where $opt_{VC}(G)$ denotes the minimum solution size for \MinVC on $G$), at least $\frac{m}{1+\varepsilon}$ edges must be covered by the result of $\mathscr{A}$. 

Thus, at most $\frac{\varepsilon}{1+\varepsilon} \cdot m$ edges are not covered. 
We now construct a set $C$ with the $k_0$ vertices selected by $\mathscr{A}$, to which we add one arbitrary vertex of each of the remaining $\frac{\varepsilon}{1+\varepsilon} \cdot m$ uncovered edges.
This set $C$ is a vertex cover for $G$, and its size is $k_0 + \frac{\varepsilon}{1+\varepsilon}  \cdot m$. 
However, in a cubic graph, $m \leq 3k_0$ (in any optimal solution for \MinVC, each vertex can cover at most 3 edges). 
Therefore, $|C| \leq k_0 + \frac{\varepsilon}{1+\varepsilon}  \cdot 3k_0 = k_0(1+\varepsilon')$, by setting $\varepsilon' = 3\frac{\varepsilon}{1+\varepsilon} $. 
This gives a PTAS for \MinVC, a contradiction to its hardness~\cite{AK00}.
\end{proof}
}

\begin{theorem}\MaxNEQCL has no PTAS (unless P$=$NP), even for neighborhood hypergraphs of graphs of maximum degree~$7$.\label{thm:apxhard}
\end{theorem}

\begin{proof}
We will give an $L$-reduction from \MaxVC (which has no PTAS, by Proposition~\ref{prop:MaxVC-APX}) to \MaxNEQCL. 
The result will then follow from Theorem~\ref{thm:L-reduc}. Given an instance $I=(G, k)$ of \MaxVC with $G=(V,E)$ a cubic graph, we construct an instance $I'$ of \MaxNEQCL that is the neighborhood hypergraph associated  to a graph $G'$  with $G'=(V',E')$ of maximum degree~$7$ in the following way. 
For each vertex $v \in V$, we create a gadget $P_v$ with twelve vertices where four among these twelve vertices  are special: they form the set $F_v = \{f_v^1,f_v^2,f_v^3,f_v^4\}$. 
The other vertices form an independent set and are adjacent to the subsets $\{f_v^4\}$, $\{f_v^2,f_v^3\}$, 
$\{f_v^1,f_v^3\}$, $\{f_v^2,f_v^4\}$, 
$\{f_v^1,f_v^4\}$, $\{f_v^1,f_v^3,f_v^4\}$, 
$\{f_v^1,f_v^2,f_v^4\}$,$\{f_v^1,f_v^2,f_v^3,f_v^4\}$, respectively. We also add edges between  $f_v^1$ and $f_v^2$, between $f_v^2$ and $f_v^3$ and between $f_v^3$ and $f_v^4$. Since $G$ is cubic, for each vertex $v$ of $G$, there are three edges $e_1$, $e_2$ and $e_3$ incident with $v$. For each edge $e_i$ ($1\leq i\leq 3$), the endpoint $v$ is replaced by $f_v^i$. Moreover, each of these original edges of $G$ is replaced in $G'$ by two edges by subdividing it once (see Figure~\ref{fig:reduction} for an illustration). We call the vertices resulting from the subdivision process, \emph{edge-vertices}.
Finally, we set $k'=4k$.


\begin{figure}[ht!]

\centering
\begin{tikzpicture}[scale=0.65]

\node[draw,circle] (f1) at (0,0) {$f_v^1$};
\node[draw,circle] (f2) at (0,-2) {$f_v^2$};
\node[draw,circle] (f3) at (0,-4) {$f_v^3$};
\node[draw,circle] (f4) at (0,-6) {$f_v^4$};

\node[draw,circle] (a4) at (-0.5,-7) {};

\node[draw,circle] (a13) at (-1.5,-2) {};
\node[draw,circle] (a14) at (-1.5,-4.2) {};
\node[draw,circle] (a23) at (-0.5,-2.8) {};
\node[draw,circle] (a24) at (-1.5,-3.1) {};
%
\node[draw,circle] (a124) at (-4,-4) {};
\node[draw,circle] (a134) at (-4,-2.5) {};
%
\node[draw,circle] (a1234) at (-5.5,-3.25) {};

\draw (f1) -- (f2) -- (f3) -- (f4) (f1) -- ++(1.5,0) (f2) -- ++(1.5,0) (f3) -- ++(1.5,0);

\draw (f4) -- (a4);
%
\draw (f1) -- (a13) -- (f3);
\draw (f1) -- (a14) -- (f4);
\draw (f2) -- (a23) -- (f3);
\draw (f2) -- (a24) -- (f4);
%
\draw (f1) -- (a124) -- (f2);
\draw (f4) -- (a124);
\draw (f1) -- (a134) -- (f3);
\draw (f4) -- (a134);
%
\draw (f1) edge[bend right] (a1234) (a1234) -- (f2);
\draw (f3) -- (a1234) edge[bend right] (f4);

\node[draw,rectangle, rounded corners,fit= (f1) (f4) (a1234) (a4)] {};
\node () at (-5.5,-8) {a)};




%



\begin{scope}[xshift=5cm,yshift=-3.5cm]

\node () at (-2,-4.5) {b)};

\node[draw,circle,minimum width=0.7cm] (v) at (-2,0) {$v$};
\node[draw,circle,minimum width=0.7cm] (u1) at (0,2) {$u_1$};
\node[draw,circle,minimum width=0.7cm] (u2) at (0,0) {$u_2$};
\node[draw,circle,minimum width=0.7cm] (u3) at (0,-2) {$u_3$};

\draw (u1) -- (v) -- (u2);
\draw (v) -- (u3) -- (u2);

\draw[line width=1.5pt,->] (1.0,0)-- ++(1,0);

\begin{scope}[xshift=3.7cm]
\node[] (fv1) at (0,2) {};
\node[] (fv2) at (0,1) {};
\node[] (dv) at (-1,1) {};
\node[] (fv3) at (0,0) {};

\node[right of=fv1, node distance=1.8cm] (fu21) {};
\node[below of=fu21, node distance=0.4cm] (fu22)  {};
\node[right of=fu21] (du2) {};
\node[below of=fu22, node distance=0.4cm] (fu23) {};

\node[below of=fu23, node distance=1.8cm] (fu31) {};
\node[below of=fu31, node distance=0.4cm] (fu32)  {};
\node[right of=fu31] (du3) {};
\node[below of=fu32, node distance=0.4cm] (fu33) {};

\node[above of=fu21, node distance=1.4cm] (fu13) {};
\node[above of=fu13, node distance=0.4cm] (fu12)  {};
\node[above of=fu12, node distance=0.4cm] (fu11) {};
\node[right of=fu12] (du1) {};

\draw (fv1) -- (fu11) node[midway,circle,draw,fill=white] {};
\draw (fv2) -- (fu21) node[midway,circle,draw,fill=white] {};
\draw (fv3) -- (fu32) node[midway,circle,draw,fill=white] {};
\draw (fu31) -- (fu23) node[midway,circle,draw,fill=white] {};

\node[draw,fill=white,rectangle, rounded corners,fit= (fv1) (fv2) (fv3) (dv)] (Pv) {$P_v$};
\node[draw,rectangle, fill=white,rounded corners,fit= (fu11) (fu12) (fu13) (du1)] (Pu1) {$P_{u_1}$};
\node[draw,fill=white,rectangle, rounded corners,fit= (fu21) (fu22) (fu23) (du2)] (Pu2) {$P_{u_2}$};
\node[draw,fill=white,rectangle, rounded corners,fit= (fu31) (fu32) (fu33) (du3)] (Pu3) {$P_{u_3}$};

\end{scope}
\end{scope}

\end{tikzpicture}
\caption{a) Vertex-gadget $P_v$ and b) illustration of the reduction.}\label{fig:reduction}
\end{figure}

  From any optimal solution $S$ with $|S|=k$ covering $opt(I)$ edges of $G$, we construct a set $C = \{ f_v^j : 1 \leqslant j \leqslant 4, v \in S  \}$ of size $4k$.  
By construction, $C$ induces $12$ equivalence classes in each vertex gadget. Moreover, for each covered edge $e=xy$ in $G$, the corresponding edge-vertex $v_e$ in $G'$ forms a class of size~$1$ (which corresponds to one or two neighbor vertices $f_x^i$ and $f_y^j$ of $v_e$ in $C$). 
Finally, all vertices in $G'$ corresponding to edges not covered by $S$ in $G$, as well as all vertices in vertex gadgets corresponding to vertices not in $S$, belong to the same equivalence class (corresponding to the empty set). 
Thus, $C$ induces in $G'$ $12k+opt(I)+1$ equivalence classes, and hence we have 
\begin{equation}
opt(I')\geq 12k+opt(I)+1.
\label{eqn-Lreduc}
 \end{equation}

Conversely, given a solution $C'$ of $I'$ with $|C'|=4k$, we transform it into a solution for $I$ as follows. 
First, we show that $C'$ can be transformed into another solution $C''$ such that (1) $C''$ only contains vertices of the form $f_v^i$, (2) each vertex-gadget contains either zero or four vertices of $C''$, and (3) $C''$ does not induce less equivalence classes than $C'$.
To prove this, we proceed step by step by locally altering $C'$ whenever (1) and (2) are not satisfied, while ensuring (3).

Suppose first that some vertex-gadget $P_v$ of $G'$ contains at least four vertices of $C'$. Then, the number of equivalence classes involving some vertex of $V(P_v)\cap C'$ is at most twelve within $P_v$ (since there are only twelve vertices in $P_v$), and at most three outside $P_v$ (since there are only three vertices not in $P_v$ adjacent to vertices in $P_v$). Therefore, we can replace $V(P_v)\cap C'$ by the four special vertices of the set $F_v$ in $P_v$; this choice also induces twelve equivalence classes within $P_v$, and does not decrease the number of induced classes.

Next, we show that it is always best to select the four special vertices of $F_v$ from some vertex-gadget (rather than having several vertex-gadgets containing less than four solution vertices each). To the contrary, assume that there are two vertex-gadgets $P_u$ and $P_v$ containing respectively $a$ and $b$ vertices of $C'$, where $1 \leq b \leq a \leq 3$. Then, we remove an arbitrary vertex from $C'\cap V(P_v)$; moreover we replace $C'\cap V(P_u)$ with the subset $\{f_u^i, 1\leq i\leq a+1\}$, and similarly we replace $C'\cap V(P_v)$ with the subset $\{f_v^i, 1\leq i\leq b-1\}$. Before this alteration, the solution vertices within $V(P_u)\cup V(P_v)$ could contribute to at most $2^{a} + 2^b-2$ equivalence classes (we substract~$2$ because we do not count the ``empty'' class towards this contribution). After the modification, if $b=1$, this quantity is at least $2^{a+1}-2$: we substract~$1$ from $2^{a+1}$ for the empty class as above, and $1$ for the singleton class that could be affected by the addition of a solution vertex to $P_v$ (this class is represented by the adjacent edge-vertex $v_e$): indeed, this class may already exist because of some other vertex-gadget also adjacent to $v_e$. Thus in that case we have at least $2^a$ additional classes after the modification. If $b\geq 2$, after the modification we have at least $2^{a+1} + 2^{b-1}-2$ (nonempty) classes related to the gadgets $P_u$ and $P_v$, from which we must substract~$1$ as before if the vertex newly added to $P_v$ does not create a new singleton class. Thus, after the modification we have at least $2^{a+1} + 2^{b-1}-3-(2^{a} + 2^b-2)=2^a-2^{b-1}-1$ additional equivalence classes, which is positive. Hence, we conclude that it can be assumed that all vertex-gadgets (except possibly at most one) contain either zero or four vertices from the solution set $C'$.

Suppose that there exists one vertex-gadget $P_v$ with $i$ solution vertices, $1\leq i\leq 3$. We show that we may add $4-i$ solution vertices to it so that $C'\cap V(P_v)=F_v$. Consider the set of edge-vertices belonging to $C'$. Since we had $|C'|=4k$ and all but one vertex-gadget contain exactly four solution vertices, there are at least $4-i$ edge-vertices in the current solution set. Then, we remove an arbitrary set of $4-i$ edge-vertices from $C'$ and instead, we replace the set $V(P_v)\cap C'$ by the set $F_v$ of special vertices of $P_v$. We now claim that this does not decrease the number of classes induced by $C'$. 
Indeed, any edge-vertex, since it has degree~$2$, may contribute to at most three equivalence classes, and the $i$ solution vertices in $P_v$ can contribute to at most $2^{i}$ classes. 
Summing up, in the old solution set, these four vertices contribute to at most $3(4-i)+2^{i}$ classes, which is less than $12$ since $1\leq i\leq 3$.
In the new solution, these four vertices contribute to at least $12$ classes, which proves our above claim. 

We now know that there are $4i$ edge-vertices in $C'$, for some $i\leq k$. All other solution vertices are special vertices in some vertex-gadgets. By similar arguments as in the previous paragraph, we may select any four of them and replace them with some set $F_v$ of special vertices of some vertex-gadget $P_v$. Before this modification, these four solution vertices may have contributed to at most $3\cdot 4=12$ classes, while the new four solution vertices now contribute to at least $12$ classes.

Applying the above arguments, we have proved the existence of the required set $C''$ that satisfies conditions (1)--(3).



Therefore, we may now assume that the solution $C''$ contains no edge-vertices, and for each vertex-gadget $P_v$, $C''\cap V(P_v)\in \{\emptyset, F_v\}$. We define as solution $S$ for $I$ the set of vertices $v$ of $G$ for which $P_v$ contains four vertices of $C''$. Then, $val(S)=val(C')-12k-1$. Considering an optimal solution $C'$ for $I'$, we have $opt(I) \geq opt(I')-12k-1$. Using~\eqref{eqn-Lreduc}, we conclude that $opt(I') = opt(I)+12k+1 \leq opt(I) + 24 opt(I) + 1$ since $k \leqslant 2opt(I)$ and thus $opt(I') \leq 26 opt(I)$.

Moreover, we have $opt(I) - val(S) = opt(I') -12k - 1 - (val(C')-12k-1) =  opt(I')-val(C')$.

Thus, our reduction is an $L$-reduction with $\alpha=26$ and $\beta = 1$.
\end{proof}

\section{Conclusion}\label{sec:conclu}

\fullversion{We have defined and studied the problems \NEQCLkl (and its maximization variant \MaxNEQCL), that generalizes \PBDS and \VCDIM.}
\confversion{In this paper, we defined and studied generalization of \PBDS and \VCDIM.}

\fullversion{It would be interesting to settle the parameterized complexity of \NEQCLkl (for the natural parameter $k$, or equivalently, $\ell$) for neighborhood hypergraphs of certain graph classes such as planar graphs, line graphs or interval graphs. Indeed the problem remains NP-hard for these cases. Also, can one extend Theorem~\ref{thm:boundeddegree} to general hypergraphs, that is, is the problem FPT when restricted to hypergraphs of bounded maximum degree or hyperedge size?}

The probably most intriguing open question seems to be the approximation complexity of \MaxNEQCL.
As a first step to solve this question, one could determine whether such an approximation algorithm exists in superpolynomial time, or on special subclasses such as neighborhood hypergraphs of specific graphs. We have seen that there exist polynomial-time approximation algorithms with a sublinear ratio for special cases; does one exist in the general case?

Finally, we remark that a recent paper~\cite{inapproxVCDIM} shows that if \MaxVCDIM admits a $(2-o(1))$-approximation algorithm in time less than $n^{\log^{1-o(1)} n}$ (with $n$ the input size), then \textsc{$3$-Sat} on $N$ variables can be solved by a Monte-Carlo type algorithm with success probability $2/3$ and running in time $2^{o(N)}$ (contradicting the "randomized Exponential Time Hypothesis" conjecture, rETH). In Theorem~\ref{thm:VCDIM-inapprox}, we prove that a $2$-approximation algorithm for \MaxNEQCL would imply a $(2-o(1))$-approximation algorithm for \MaxVCDIM (with randomized polynomial time overhead). Unfortunately, our $o(1)$ seems to be smaller than the $o(1)$ from~\cite{inapproxVCDIM}, therefore we cannot apply their result directly. Thus, we leave it as an open problem whether one can rule out a $2$-approximation algorithm for \MaxNEQCL assuming the rETH (or some other complexity theory conjecture).



\end{document}